\documentclass[lettersize,journal]{IEEEtran}
\usepackage{algorithmic}
\usepackage{algorithm}
\usepackage{array}
\usepackage[caption=false,font=normalsize,labelfont=sf,textfont=sf]{subfig}
\usepackage{textcomp}
\usepackage{stfloats}
\usepackage{url}
\usepackage{verbatim}

\usepackage{cite}
\usepackage{amsmath,amssymb,amsfonts}
\usepackage{amsthm}
\newtheorem{remark}{Remark}

\newtheorem{theorem}{Theorem}

\newtheorem{definition}{Definition}
\usepackage{algorithmic}
\usepackage{graphicx}
\usepackage{textcomp}
\usepackage{xcolor}
\usepackage{mathrsfs}
\def\BibTeX{{\rm B\kern-.05em{\sc i\kern-.025em b}\kern-.08em
    T\kern-.1667em\lower.7ex\hbox{E}\kern-.125emX}}
\usepackage{amssymb}
\usepackage{float}
\usepackage{tikz}

\allowdisplaybreaks

\begin{document}

\title{Channel-coded Over-the-Air Computation}

\author{Shudi Weng,~\IEEEmembership{Graduate Student Member,~IEEE,}
Ming Xiao,~\IEEEmembership{Senior Member,~IEEE,} and \\
Mikael Skoglund,~\IEEEmembership{Fellow,~IEEE}\vspace{-2em}
\thanks{

Shudi Weng, Ming Xiao, and Mikael Skoglund are with the Department of Information Science and Engineering (ISE), KTH Royal Institute of Technology, Stockholm 11428, Sweden (e-mail: \{shudiw, mingx, skoglund\}@kth.se).}
}

\markboth{Journal of \LaTeX\ Class Files,~Vol.~14, No.~8, August~2021}%
{Shell \MakeLowercase{\textit{et al.}}: A Sample Article Using IEEEtran.cls for IEEE Journals}


\maketitle
\begin{abstract}
AirComp enables efficient wireless data aggregation, however, its computation accuracy is vulnerable to channel impairments. 
As a promising solution, the role of redundancy for AirComp has been largely unexplored, leaving a critical gap in reliable AirComp design. To address this, we propose a channel-coding scheme tailored for AirComp that preserves the aggregation structure while mitigating channel impairments through bandwidth expansion. We show that the computation error decreases with the coding rate and can vanish. Both theoretical and simulation results demonstrate that the proposed scheme significantly improves computation performance.
\end{abstract}

\begin{IEEEkeywords}
Over-the-air computation, Channel Coding. 
\end{IEEEkeywords}

\vspace{-0.5em}
\section{Introduction}
\IEEEPARstart{O}{ver}-the-air computation (AirComp) enables the computation of mathematical functions by exploiting the waveform superposition property of electromagnetic waves over multiple-access channels (MACs) \cite{csahin2023survey}, significantly improving communication efficiency. 
This advantages highlight its large potential in 6G, where various intelligent networks demand extensive computation over wireless networks.
However, computation accuracy, the key performance metric in AirComp, is sensitive to channel impairments due to, e.g., noise and fading, posing a fundamental challenge to AirComp. 

To mitigate computation errors in AirComp, recent works have mainly focused on signal processing and transmission design under existing channel uses rather than introducing redundancy.
For instance, \cite{9834026} studies transceiver design, \cite{jing2023transceiver,zhang2025beamforming} focus on beamforming, and \cite{10893712} uses bit-slicing multi-symbol transmission.
While these approaches improve AirComp to a certain extent, they remain fundamentally limited as they lack further error correction mechanisms.

In classic communication systems, the role of redundancy introduced by channel coding has been extensively studied under both orthogonal and non-orthogonal paradigms. In orthogonal transmission, reliability is enhanced by using a longer codeword to combat channel imperfections. 
In non-orthogonal communications, strategies such as compute-and-forward~\cite{nazer2011compute} exploit structured channel codes and interference to improve transmission reliability and multi-user efficiency.
However, these approaches aim at individual message recovery, which fundamentally differs from the objective of AirComp to compute a deterministic function via in-air signal superposition. 
This mismatch in transmission objectives renders most existing channel coding strategies not suitable for AirComp.

Pioneering work \cite{4305404} laid the foundation for computation code over a multiple-access channel (MAC). 
Its most renowned contribution is the use of structured lattice codes to introduce redundancy and exploit the alignment between lattice linearity and channel superposition for function computation, under a linear mapping from source messages to lattice codewords.
Building on this, \cite{9895450,11358822} further improve AirComp using lattice-based coding incorporating other techniques, e.g., stochastic quantization and low-density parity-check (LDPC) codes.
Another existing strategy is the repetition coding \cite{10622499,11303856}. 
Beyond these efforts, channel coding design tailored to AirComp is in its infancy, highlighting a critical gap in achieving reliable AirComp under channel impairments.

\vspace{-0.5em}
\subsection{Contributions}
This work advances coded AirComp beyond existing lattice-coding based paradigms and applies to both digital and analog AirComp.
The main contributions are summarized as follows:
\begin{itemize}
\item We propose a channel-coded AirComp protocol in which all users apply an identical encoding matrix to preserve the summation structure while mitigating fading and noise. Although our scheme and the achievable scheme in \cite[Theorem 1]{4305404} both use a common coding matrix, they have different objectives and apply to different system models.
\item We characterize the statistical properties of the computation error for the proposed channel-coded AirComp and derive the optimal coding design and its construction.
\item We characterize the achievable rate regions of the proposed channel-coded AirComp under various computation error criteria.
\end{itemize}

\vspace{-0.5em}
\subsection{Notations}
Non-bold letters denote scalars, bold lowercase letters denote vectors, and bold uppercase letters denote matrices. The superscript $(\cdot)^\mathsf{H}$ denotes the Hermitian transpose, $(\cdot)^{-1}$ denotes the matrix inverse, $(\cdot)^\dagger$ denotes the Moore–Penrose pseudoinverse, and $\otimes$ denotes Kronecker product. The notation $\mathcal{CN}$ represents the complex Gaussian distribution, $\mathcal{E}$ denotes the exponential distribution, and $\Gamma$ denotes the Gamma distribution. The notation $\mathbf{1}_I $ denotes $I\times 1$ all-one vector, $\mathbf{I}_L$ denotes the $L \times L$ identity matrix. The set $\mathbb{C}$ represents the complex domain. The operator $\mathbb{E}[\cdot]$ denotes expectation, $\mathrm{cov}(\cdot)$ denotes covariance, $\mathrm{var}(\cdot)$ denotes variance, $\mathrm{Pr}(\cdot)$ denotes probability, $\lvert \cdot \rvert$ denotes absolute value, $\lVert \cdot \rVert$ denotes $l_2$-norm, $\mathrm{tr}(\cdot)$ denotes the trace operator, $\mathrm{diag}(\cdot)$ denotes diagonal extraction operator, and $\mathrm{dim}(\cdot)$ denotes the dimension of a vector. The symbol $\overset{\mathrm{d}}{=}$ denotes equality in distribution.
\begin{figure*}[bt]
    \centering
    \includegraphics[width=0.8\linewidth]{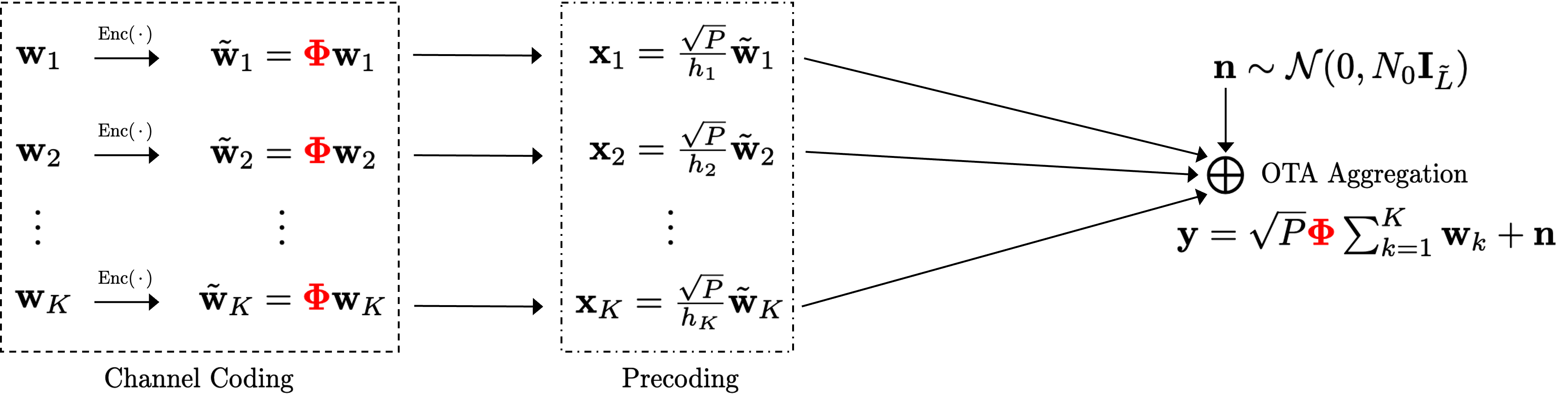}
    \caption{Overview of the proposed channel-coded AirComp scheme, in which an identical encoding matrix is applied by all users.}
    \vspace{-2mm}
    \label{fig: ChannelOTA}
\end{figure*}
\section{System Model and Preliminaries}
\subsection{Uncoded AirComp Protocol}
Consider an AirComp system consisting of $K$ users and a base station (BS). Each user \(k\) holds a source message \(\mathbf{w}_k \in \mathbb{C}^L\) to be transmitted to the BS. The BS aims to compute the sum $\mathbf{w}=\sum_{k\in[K]} \mathbf{w}_k$. 
Assume that the entries in $\mathbf{w}_k$ are identically and independently distributed (i.i.d.) with a per-letter power constraint $\mathbf{w}_k\sim\mathcal{CN}(0, P_W\mathbf{I}_L)$.
Let \(h_k \in \mathbb{C}\) denote the uplink channel coefficient from user \(k\) to the BS, and let $\mathbf{x}_k\in\mathbb{C}^L$ denote the transmit signal. 
Assume perfect channel state information at the transmitter (CSIT), using channel inversion precoding, user $k$ transmits
\begin{align}
    \mathbf{x}_k = \frac{\sqrt{P}}{h_k}\mathbf{w}_k,
    \qquad h_k \neq 0,
    \label{eq: tx_signal}
\end{align}
where \(P>0\) is a common power scaling factor and $\alpha_{k}\triangleq\frac{\sqrt{P}}{h_k}$ is the transmit scaling factor. Moreover, the average transmit power of client \(k\) is assumed to satisfy a power constraint, 
\begin{align}
    \frac{1}{\mathrm{dim}(\mathbf{x}_k)}\mathbb{E}\!\left[\|\mathbf{x}_k\|^2\right]
    = \frac{P}{|h_k|^2 \mathrm{dim}(\mathbf{x}_k)}\mathbb{E}\!\left[\|\mathbf{w}_k\|^2\right]
    \le P_X.
    \label{eq: power_constraint}
\end{align}
Thus, the signal received at the BS is
\begin{align}
    \mathbf{y}
    &= \sum_{k=1}^K h_k \mathbf{x}_k + \mathbf{n} \label{eq: y_signal1}\\
    &= \sqrt{P}\sum_{k=1}^K \mathbf{w}_k + \mathbf{n},
    \label{eq: y_signal2}
\end{align}
where the noise vector \(\mathbf{n} \in \mathbb{C}^L\) is independent of the transmit signal and follows $ \mathbf{n} \sim \mathcal{CN}(\mathbf{0}, N_0 \mathbf{I}_{\mathrm{dim}(\mathbf{x}_k)})$.
\subsection{Performance Metrics}
\begin{definition}[Coding rate]
The rate $R_k$ of a channel coding scheme at transmitter $k$ is defined as the ratio between the source message length $L_k$ (measured in dimension) and the number of channel uses $\tilde{L}_k$, i.e.,
\begin{align}
    R_k \triangleq \frac{L_k}{\tilde{L}_k}.
\end{align}
In this work, we consider a symmetric setting where $R_1 = \cdots = R_K = \frac{L}{\tilde{L}} \triangleq R$.
\end{definition}
\begin{definition}[Computation error]
The computation error of an AirComp protocol is defined as the expected distortion between the true sum $\mathbf{w}$ and the estimated sum $\hat{\mathbf{w}}$, i.e.,
\begin{align}
    \gamma \triangleq \mathbb{E}\!\left[ d(\mathbf{w}, \hat{\mathbf{w}}) \right],
\end{align}
where $d(\cdot,\cdot)$ denotes a distortion measure. In this work, we adopt the mean square error (MSE).
\end{definition}

\vspace{-0.5em}
\section{Proposed Channel-coded AirComp} \label{sec: proposed}
The proposed channel-coded AirComp scheme is illustrated in Fig. \ref{fig: ChannelOTA}. Each user first encodes its source message $\mathbf{w}_k\in \mathbb{C}^{L}$ into a channel codeword $\Tilde{\mathbf{w}}_k\in\mathbb{C}^{\Tilde{L}}$ before transmission, i.e., $\Tilde{\mathbf{w}}_k=\mathrm{Enc}(\mathbf{w}_k)$, where $\mathrm{Enc}: \mathbb{C}^{L}\rightarrow \mathbb{C}^{\Tilde{L}}$. To enable accurate recovery of the aggregated signal, the encoding dimension must satisfy $\Tilde{L}\geq L$, i.e., $R\leq 1$. 

To enable accurate recovery of the aggregated signal after in-air superposition, one approach is to equip all users with an identical encoding matrix $\boldsymbol{\Phi}\in \mathbb{C}^{\Tilde{L}\times L}$. 
The encoding matrix is designed to satisfy two key properties:
(i) it preserves the total transmit power, ensuring that the encoded signals have the same total power as the original uncoded messages, so that any performance improvement is solely attributed to coding gain; and
(ii) any subset of $L$ rows of $\boldsymbol{\Phi}$ forms a full-rank submatrix, which guarantees decoding feasibility, i.e.,
\begin{subequations}
\begin{align}
&\mathrm{tr}(\boldsymbol{\Phi}^\mathsf{H} \boldsymbol{\Phi})=L,\\
&\forall\, \mathcal{S} \subseteq [\tilde{L}],\ |\mathcal{S}| = L,\quad
\operatorname{rank}(\boldsymbol{\Phi}_{\mathcal{S},:}) = L.\label{eq: rank}
\end{align}
\label{eq: Phi_constraint}
\end{subequations}
At user $k$, the channel codeword is generated by
\begin{align}
    \Tilde{\mathbf{w}}_k= \boldsymbol{\Phi} \mathbf{w}_k.
\end{align}
It can be verified that $\mathbb{E}[\lVert \boldsymbol{\Phi} \mathbf{w}_k \rVert^2]=\mathbb{E}[\lVert \mathbf{w}_k \rVert^2]=LP_W$. 
Then, each user generates a transmit signal $\mathbf{x}_k$ using channel inversion precoding as in \eqref{eq: tx_signal}, and transmit $\mathbf{x}_k$ to the BS,
\begin{align}
    \mathbf{x}_k = \frac{\sqrt{P}}{h_k}\Tilde{\mathbf{w}}_k,
    \qquad h_k \neq 0.
    \label{eq: tx_signal_channel_coding}
\end{align}
Following \eqref{eq: y_signal1}, the BS receives 
\begin{align}
    \mathbf{y}
    &= \sqrt{P} \boldsymbol{\Phi} \sum_{k=1}^K \mathbf{w}_k + \mathbf{n}.
\end{align}
Thus, we logically transfer an MAC channel into a single-user channel. Accordingly, the BS can decode the sum as 
\begin{align}
    \hat{\mathbf{w}}=\frac{\boldsymbol{\Phi}^{\dagger}\mathbf{y}}{\sqrt{P}}=\sum_{k=1}^K \mathbf{w}_k+ \frac{\boldsymbol{\Phi}^{\dagger}\mathbf{n}}{\sqrt{P}}. 
    \label{eq: final_est}
\end{align}
Following \eqref{eq: power_constraint}, the transmission power is constrained by
\begin{align}
    &\frac{1}{\Tilde{L}} \mathbb{E}\left[ \left\lVert \frac{\sqrt{P}}{h_k} \mathbf{\Phi} \mathbf{w}_k \right\rVert^2 \right]
    =\frac{L P P_W}{\Tilde{L}\lvert h_k\rvert^2}
    \leq P_X, \quad \forall k.
    \label{eq: coding_power_constraint}
\end{align}
Accordingly, the common power scaling factor should satisfy $P\leq \frac{P_X \min_k\{\lvert h_k\rvert^2 \}}{R P_W}$. 


\vspace{-0.5em}
\section{Distortion Analysis} \label{sec: Distortion Analysis}
This section characterizes the statistical property of the resulting MSE in a closed form. For simplicity of notations, define the transmission signal-to-noise ratio (SNR) by $\rho_X\triangleq\frac{P_X}{N_0}$, and the normalized SNR by $\rho\triangleq\frac{P}{N_0}$.  
\vspace{-0.5em}
\subsection{MSE in Expectation}
Define the effective noise in the final estimation in \eqref{eq: final_est} by $\mathbf{n}_{\mathrm{eff}}\triangleq\frac{\boldsymbol{\Phi}^{\dagger}\mathbf{n}}{\sqrt{P}}$, $\mathbf{n}_{\mathrm{eff}}\in \mathbb{C}^L$. Then, we have
\begin{align}
    \mathrm{cov}(\mathbf{n}_{\mathrm{eff}})
    &\hspace{-0.5mm}= \hspace{-0.5mm}\frac{1}{P}\boldsymbol{\Phi}^{\dagger}
    \mathrm{cov}(\mathbf{n})
    (\boldsymbol{\Phi}^{\dagger})^\mathsf{H} 
    \hspace{-1mm}= \hspace{-0.5mm}\frac{N_0}{P}\hspace{-0.5mm} \boldsymbol{\Phi}^{\dagger}
    (\boldsymbol{\Phi}^{\dagger})^\mathsf{H}
    \hspace{-1mm}=\hspace{-0.5mm}\frac{1}{\rho} (\boldsymbol{\Phi}^\mathsf{H} \boldsymbol{\Phi} )^{-1}\hspace{-0.5mm}.
    \label{eq: cov_effect_noise}
\end{align}
Thus, the computation error measured by the MSE of the proposed channel-coded AirComp is given by
\begin{align}
    d(\mathbf{w}, \hat{\mathbf{w}})=\frac{1}{L} \lVert \hat{\mathbf{w}}-\mathbf{w} \rVert^2=\frac{1}{L} \lVert \mathbf{n}_{\mathrm{eff}} \rVert^2,
\end{align}
and its expectation is 
\begin{align}
    \gamma = \mathbb{E}\left[ d(\mathbf{w}, \hat{\mathbf{w}}) \right]
    =\frac{1}{L\rho} \mathrm{tr}\left( (\boldsymbol{\Phi}^\mathsf{H} \boldsymbol{\Phi} )^{-1} \right), \label{eq: mse} 
\end{align}
where the last equality in \eqref{eq: mse} applies \eqref{eq: cov_effect_noise}. 
\vspace{-0.5em}
\subsection{MSE in Distribution } \label{sec:mse_distribution}
Since $\boldsymbol{\Phi}^\mathsf{H} \boldsymbol{\Phi}$ is Hermitian positive definite, it admits the eigenvalue decomposition $\boldsymbol{\Phi}^\mathsf{H}\boldsymbol{\Phi} = \mathbf{U} \boldsymbol{\Lambda} \mathbf{U}^\mathsf{H}$, where $\mathbf{U}$ is unitary and $\boldsymbol{\Lambda}=\mathrm{diag}(\lambda_1,\dots,\lambda_L)$. Accordingly, $(\boldsymbol{\Phi}^\mathsf{H}\boldsymbol{\Phi})^{-1} = \mathbf{U} \boldsymbol{\Lambda}^{-1} \mathbf{U}^\mathsf{H}$, where $\boldsymbol{\Lambda}^{-1}=\mathrm{diag}(\frac{1}{\lambda_1},\dots,\frac{1}{\lambda_L})$. Therefore,
\begin{align}
    \mathrm{cov}(\mathbf{n}_{\mathrm{eff}})=\frac{1}{\rho}
    \mathbf{U} \boldsymbol{\Lambda}^{-1} \mathbf{U}^\mathsf{H}.
\end{align}
Following this, $\mathbf{n}_{\mathrm{eff}}$ can be rewritten as
\begin{align}
    \mathbf{n}_{\mathrm{eff}}&\overset{\mathrm{d}}{=}\sqrt{\frac{1}{\rho}}  \mathbf{U} \boldsymbol{\Lambda}^{-\frac{1}{2}} \mathbf{z},
\end{align}
where $\mathbf{z}\sim\mathcal{CN}(\mathbf{0}, \mathbf{I}_L)$. Since $\mathbf{U}$ is unitary and $\mathbf{z}$ is circularly symmetric, $ \mathbf{U}\mathbf{z}\overset{\mathrm{d}}{=}\mathbf{z}$. It follows that 
\begin{align}
    \frac{1}{L}\lVert\mathbf{n}_{\mathrm{eff}}\rVert^2 &\overset{\mathrm{d}}{=}\frac{1}{\rho L} \mathbf{z}^\mathsf{H} \boldsymbol{\Lambda}^{-1} \mathbf{z}\overset{\mathrm{d}}{=}\frac{1}{\rho L}\sum_{l=1}^L \frac{1}{\lambda_l} \lvert z_l \rvert^2, 
    \label{eq: general_dist_pdf}
\end{align}
where $\lvert z_1 \rvert^2, \cdots, \lvert z_L \rvert^2 \overset{\mathrm{i.i.d.}}{\sim}\mathcal{E}(1)$.

\vspace{-0.5em}
\section{Optimal Coding Scheme}\label{sec: optimal scheme}
This section presents the optimal channel coding scheme, including its design principles and construction.  
\vspace{-0.5em}
\subsection{Optimal Coding Scheme}
Given a source message length $L$ and a coding rate $R$, the optimal coding scheme is defined as the encoding matrix $\mathbf{\Phi}$ that minimizes MSE in \eqref{eq: mse} while ensuring decoding feasibility in \eqref{eq: Phi_constraint} and satisfying the power constraint in \eqref{eq: coding_power_constraint}. This leads us to the following optimization problem:
\begin{align}
   (\textbf{P1})\quad &\min_{\mathbf{\Phi}, P}\quad \eqref{eq: mse}, \label{eq: opt_prob1}\\
    &\hspace{1mm}\mathrm{s.t.} \quad \eqref{eq: coding_power_constraint} \quad \mathrm{and} \quad \eqref{eq: Phi_constraint}.\notag
\end{align}
By \eqref{eq: mse} and \eqref{eq: coding_power_constraint}, the MSE is minimized by maximizing $\rho$, yielding the optimal solution 
$\rho^*\triangleq\frac{P^*}{N_0}=\frac{\rho_X \min_k\{\lvert h_k\rvert^2 \}}{R P_W}$ and $P^*\triangleq \frac{P_X \min_k\{\lvert h_k\rvert^2 \}}{R P_W}$. By the property of $\boldsymbol{\Phi}^\mathsf{H} \boldsymbol{\Phi}$ discussed in the beginning of Section \ref{sec:mse_distribution}, (\textbf{P1}) can be simplified to
 \vspace{-1mm}
\begin{align}
       (\textbf{P2})\quad &\min_{\mathbf{\Phi}}\quad \sum_{l=1}^L \frac{1}{\lambda_l}, \label{eq: opt_prob2}\\
    &\hspace{1mm}\mathrm{s.t.} \quad \sum_{l=1}^L \lambda_l =L, \quad \mathrm{and} \quad \eqref{eq: rank}. \notag
    \vspace{-1mm}
\end{align}
Solving (\textbf{P2}) leads us to the following theorem.
\begin{theorem}(Optimal coding scheme) \label{theo: Opt_design}
Given a source message length $L$ and a coding rate $R=\frac{L}{\Tilde{L}}$, the optimal coding scheme is achieved by any
$\boldsymbol{\Phi} \in \mathbb{C}^{\tilde{L} \times L}$ satisfying \eqref{eq: rank} and 
\begin{align}
    \boldsymbol{\Phi}^\mathsf{H} \boldsymbol{\Phi} =  \mathbf{I}_L. 
    \label{eq: opt_principle}
\end{align}
The MSE in expectation of the optimal coding scheme is 
\begin{align}
    \gamma_{\mathrm{opt}}=\frac{1}{\rho^*}= \frac{ R P_W}{ \rho_X \min_k\left\{ \lvert h_k\rvert^2 \right\}}. 
    \label{eq: mse_opt}
\end{align}
The MSE follows $\Gamma$-distribution, 
\begin{align}
d(\mathbf{w}, \hat{\mathbf{w}})\sim \Gamma \left(L, \frac{P_W}{\Tilde{L}\rho_X \min_k\{\lvert h_k\rvert^2 \}}\right). 
\label{eq: opt_mse_distr}
\end{align} 
\end{theorem}
\begin{proof}[Proof Sketch]
Consider $(\textbf{P2})$. Since $f(x)=\frac{1}{x}$ is convex for $x>0$, by Jensen's inequality, $\frac{1}{L} \sum_{l=1}^L \frac{1}{\lambda_l}\geq  \frac{1}{\frac{1}{L} \sum_{l=1}^L\lambda_l}\geq 1$.
The equality holds if and only if $\lambda_l=1$ for $\forall l$. 
When $\boldsymbol{\Lambda}=\mathbf{I}_L$, we have $\boldsymbol{\Phi}^\mathsf{H}\boldsymbol{\Phi} 
    = \mathbf{U} \left(\mathbf{I}_L\right) \mathbf{U}^\mathsf{H}
    = \mathbf{U}\mathbf{U}^\mathsf{H}
    = \mathbf{I}_L$.
This completes the proof of \eqref{eq: opt_principle}. Combine this with (\textbf{P1}), \eqref{eq: mse_opt} is derived. 
Substitute $\lambda_l=1$ and $P^*$ into \eqref{eq: general_dist_pdf}, we obtain $\frac{1}{L}\lVert\mathbf{n}_{\mathrm{eff}}\rVert^2 \overset{\mathrm{d}}{=}\frac{1}{\rho^*}\sum_{l=1}^L
\lvert z_l \rvert^2 \sim \Gamma\left( L, \frac{1}{L\rho^*} \right)$ and completes \eqref{eq: opt_mse_distr}. 
\end{proof}
\begin{remark}
From \eqref{eq: mse_opt}, the computation error decreases as the coding rate $R$ reduces, as expected. If the channel gains are uniformly bounded, i.e., $\lvert h_k\rvert^2 \geq \varepsilon > 0$, then $\gamma_{\mathrm{opt}} \rightarrow 0$ as $R \rightarrow 0$. 
Notably, the expected MSE distortion is not directly dependent on $\tilde{L}$. However, a larger $\tilde{L}$ reduces the variance of the MSE, thereby improving its concentration around its mean.
\end{remark}
\begin{remark}
Notably, the $I$-fold repetition scheme \cite{10622499}, where $I$ is restricted to an integer, is a special case of the proposed method under perfect CSIT, corresponding to $\boldsymbol{\Phi}=\mathbf{1}_I^\top \otimes \mathbf{I}_L$ and $R=1/I$. Compared with the repetition-based scheme, our method has two advantages: \textit{i)} under perfect CSIT, it provides a generalized redundancy mechanism and allows for a flexible coding rate; and \textit{ii)} under imperfect CSIT, it enables flexible coding matrix design that achieves optimality.
\end{remark}

\subsection{Optimal Coding Matrix Construction}
One way to construct such an optimal channel coding matrix $\boldsymbol{\Phi} \in \mathbb{C}^{\tilde{L}\times L}$ in Theorem \ref{theo: Opt_design} is via random orthogonalization. 
Let $\mathbf{A} \in \mathbb{C}^{\tilde{L}\times L}$ be a random matrix whose entries are drawn independently 
from a standard complex Gaussian distribution $\mathcal{CN}(0,1)$. Compute the QR decomposition $\mathbf{A} = \mathbf{Q}\mathbf{R}$, where 
$\mathbf{Q} \in \mathbb{C}^{\tilde{L}\times L}$ has orthonormal columns, i.e., $\mathbf{Q}^\mathsf{H} \mathbf{Q} = \mathbf{I}_L$. 
Set
\begin{align}
    \boldsymbol{\Phi} = \mathbf{Q}.
\end{align}
By random construction, $\boldsymbol{\Phi}$ satisfies rank condition \eqref{eq: rank} amost surely. Moreover, $\boldsymbol{\Phi}^\mathsf{H} \boldsymbol{\Phi}
    =\mathbf{Q}^\mathsf{H} \mathbf{Q}
    =\mathbf{I}_L$,
achieving the optimal coding design. 

\vspace{-0.5em}
\section{Achievable Rate Regions} 
Next, we derive the achievable rate regions of the optimal coding scheme under various distortion criteria.
\begin{definition}[$\epsilon$-accurate in expectation]\label{def:e_accurate}
For a given channel realization $\mathbf{h}=(h_1,\dots,h_K)$ and a fixed channel-coded AirComp scheme, the $\epsilon$-accurate rate region $\mathcal{R}_\epsilon(\mathbf{h})$ is defined as the set of all rate tuples $(R_1,\dots,R_K)$ such that the resulting computation error satisfies
\begin{equation}
    \gamma \le \epsilon.
\end{equation}
\end{definition}

\begin{theorem}[$\epsilon$-accurate rate region]
For a given channel realization $\mathbf{h}=(h_1,\dots,h_K)$, the $\epsilon$-accurate rate region achieved by the optimal coding scheme in AirComp is
\begin{align}
\mathcal{R}_\epsilon(\mathbf{h})
=\left\{R: 0\leq R\leq \min\left\{ 1, \frac{\epsilon \rho_X \min_{k} \{\lvert h_k\rvert^2\}}{P_W}\right\}\right\}. 
\end{align}
\end{theorem}
\begin{proof}[Proof Sketch]
    Let $\gamma_{\mathrm{opt}}\leq \epsilon$, $\mathcal{R}_\epsilon(\mathbf{h})$ is derived. 
\end{proof}

\begin{definition}[$\epsilon$-asymptotically accurate]
For a given channel realization $\mathbf{h}=(h_1,\dots,h_K)$ and a fixed channel-coded AirComp scheme, 
the $\epsilon$-accurate asymptotic rate region $\mathcal{R}_\epsilon^\infty(\mathbf{h})$ is defined as the set of all rate tuples $(R_1,\dots,R_K)$, where $R_k = \frac{L_k}{\tilde{L}_k}$, such that
\begin{equation}
\limsup_{\Tilde{L}\to\infty} d(\mathbf{w},\Hat{\mathbf{w}}) \le \epsilon. 
\end{equation}
\end{definition}

\begin{theorem}[$\epsilon$-asymptotically accurate rate region]  \label{theo: asymptot_rate}
For a given channel realization $\mathbf{h}=(h_1,\dots,h_K)$, if the channel gains are uniformly bounded, i.e., $\lvert h_k\rvert^2\geq \varepsilon> 0$,  and $\Tilde{L}$ is sufficiently large, the $\epsilon$-asymptotically accurate rate region of the optimal coding scheme is characterized by
\begin{align}
    \mathcal{R}^\infty_\epsilon(\mathbf{h})
=\left\{R: 0\leq R\leq \min\left\{ 1, \frac{\epsilon \rho_X \min_{k} \{\lvert h_k\rvert^2\}}{P_W}\right\}\right\}.
\end{align}
\end{theorem}
\begin{proof}[Proof Sketch]
From \eqref{eq: opt_mse_distr}, we get
\begin{align}
    \mathrm{var}\left(d(\mathbf{w}, \hat{\mathbf{w}})\right)=L\cdot \left(\frac{P_W}{\Tilde{L}\rho_X\min_k\{\lvert h_k\rvert^2 \}}\right)^2. 
\end{align}
Fix $R$, when $\Tilde{L}\rightarrow \infty$, $\mathrm{var}\left(d(\mathbf{w}, \hat{\mathbf{w}})\right)\rightarrow 0$, hence $d(\mathbf{w}, \hat{\mathbf{w}})\rightarrow \mathbb{E}[d(\mathbf{w}, \hat{\mathbf{w}})]=\gamma_{\mathrm{opt}}$. Let $\gamma_{\mathrm{opt}}\leq \epsilon$, $ \mathcal{R}^\infty_\epsilon(\mathbf{h})$ is derived.  
\end{proof}

\begin{definition}[$(\epsilon,\delta)$-accurate]
\label{def:e_del_accurate}
For a given channel realization $\mathbf{h}=(h_1,\dots,h_K)$, a fixed channel-coded AirComp scheme, and a fixed codeword length $\tilde{L}_k$, the $(\epsilon,\delta)$-accurate rate region $\mathcal{R}_{\epsilon}^{\delta}(\mathbf{h})$ is defined as the set of all rate tuples $(R_1,\dots,R_K)$ such that
\begin{equation}
\Pr\!\left( d(\mathbf{w}, \hat{\mathbf{w}}) \le \epsilon \right) \ge 1-\delta. 
\end{equation}
\end{definition}
\begin{theorem}[$(\epsilon, \delta)$-accurate rate region]
\label{theo: Opt_ep_del}
Consider a given channel realization $\mathbf{h}=(h_1,\dots,h_K)$.
Let $\eta>0$ and define distortion threshold $\epsilon=(1+\eta)\gamma_{\mathrm{opt}}$. If the source message length $L$ satisfies 
\begin{align}
    L\geq \frac{\ln (1/\delta)}{\eta-\ln(1+\eta)},
    \label{eq: souce_length}
\end{align}
by adopting the optimal coding scheme, the following achievable $(\epsilon, \delta)$-accurate rate region establishes, 
\begin{align}
    \mathcal{R}_{\epsilon}^{\delta}(\mathbf{h})
    = \left\{R: 0\leq R\leq \min\left\{ 1, \frac{\epsilon \rho_X \min_{k} \{\lvert h_k\rvert^2\}}{(1+\eta)P_W }\right\}\right\}.
    \label{eq: rate_ep_del}
\end{align}
\end{theorem}
\begin{proof}[Proof Sketch]
For a Gamma distribution $x\sim\Gamma(\alpha, \theta)$, the Chernoff upper tail bound is 
\begin{align}
    \mathrm{Pr}(x\geq (1+\eta)\mu)\leq \mathrm{exp}(-\alpha(\eta-\ln(1+\eta))). 
    \label{eq: chernoff}
\end{align}
For the MSE distribution in \eqref{eq: opt_mse_distr}, $\mu=\gamma_{opt}$, $\alpha=L$. If $\epsilon\geq (1+\eta)\gamma_{opt}$ and $\delta\geq \mathrm{exp}(-L(\eta-\ln(1+\eta)))$, by \eqref{eq: chernoff}, $\mathrm{Pr}(x\leq \epsilon)\geq 1-\delta$. This completes the proof. 
\end{proof}
It is worth noting that the $(\epsilon, \delta)$-accurate rate region is solely defined by \eqref{eq: rate_ep_del}, which follows from the expected MSE constraint. The condition on $L$ in \eqref{eq: souce_length} does not affect the rate region, instead, 
it is used to control the concentration rate of the distortion and ensure that the Chernoff bound is sufficiently tight. That is, \eqref{eq: souce_length} guarantees that the distortion concentrates around its mean sufficiently sharp so that the $(\epsilon,\delta)$-accuracy requirement is satisfied, i.e., $\Pr(d(\mathbf{w}, \hat{\mathbf{w}}) \le \epsilon) \geq 1 - \delta$.

\vspace{-0.5em}
\section{Simulations}
This section presents numerical results comparing the proposed scheme with state-of-the-art approaches under perfect CSIT:
\begin{itemize}
    \item Proposed channel-coded AirComp adopting the optimal coding scheme.
    \item Uncoded analog AirComp.
    \item Nested lattice coding-based AirComp in \cite{9895450}, where stochastic quantization is applied to the analog source prior to transmission, and the quantization level is maximized subject to the achievable rate constraint. The parameter $\Tilde{L}$ is set as $\Tilde{L} = R L$.  
\end{itemize}
Without other specifications, the system parameters are set as follows: $K=10$, $L=5$, $P_W=1$. 
The channel coefficients are modeled as Rician fading to capture both line-of-sight (LoS) and scattered components. 
The channel between user $k$ and the receiver is modeled by $h_k = \sqrt{\frac{\kappa}{\kappa+1}} + \sqrt{\frac{1}{\kappa+1}}\iota $
where the Rician factor is set to $\kappa=5 \mathrm{dB}$ and $\iota \sim \mathcal{CN}(0,1)$. The additive channel noise is modeled as $\mathbf{z} \sim \mathcal{CN}(\mathbf{0}, N_0 \mathbf{I}_D)$ with $N_0=1$. 
All users are subject to a common transmit SNR constraint, $\rho_X \leq \mathrm{SNR\;threshold}$, and the scaling factor $P$ is chosen accordingly to ensure feasibility.

\vspace{-0.5em}
\subsection{Comparison with Benchmarks}
Fig.~\ref{fig:dist-snr} shows that the proposed channel-coded AirComp constantly outperforms both uncoded and nested lattice-based methods, with performance improving as the coding rate $R$ decreases. Moreover, the MSE decreases steadily with increasing SNR, which is consistent with Theorem~\ref{theo: Opt_design}. 
This advantage is attributed to the stronger averaging effect resulting from the increased redundancy, with the averaging performed according to the coding structure.
In contrast, the benchmark methods exhibit limited or trivial performance gains. The uncoded analog AirComp improves with SNR, but only achieves acceptable distortion at sufficiently high SNR. 
The nested lattice-based method, restricted to digital AirComp, requires prior quantization of the analog source. However, the achievable rate is inherently limited by the constant term $\frac{1}{L}$ in \cite[(17)]{9895450} that does not diminish with SNR, which in turn restricts the quantization accuracy.
Thus, even with sufficient transmit power, the computation accuracy remains limited by the source distortion, which is bounded by $\frac{N_0}{L}$ \cite[(14)]{9895450} and does not decrease with increasing codeword length.

\begin{figure}
    \centering
    \begin{minipage}[b]{0.5\textwidth}
        \centering
        \includegraphics[width=0.99\linewidth]{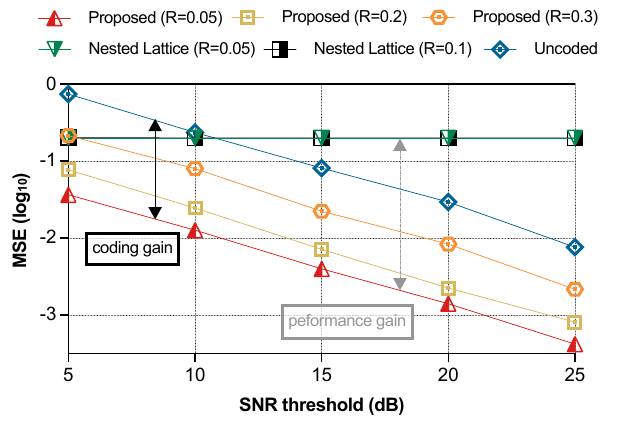}
        \vspace{-7.5mm}
        \caption{MSE comparison of the proposed method with benchmark methods under different rates and SNR constraints.}  
        \vspace{-3mm}
        \label{fig:dist-snr}
    \end{minipage}
\end{figure}

\vspace{-0.5em}
\subsection{Rate Regions}
Fig.~\ref{fig:rate-region} illustrates the achievable rate regions under MSE constraints under both asymptotically infinite and finite blocklength regimes, assuming $\min_k\{\lvert h_k\rvert^2 \}=1$. It can be observed that increasing the SNR expands the achievable rate region, while stricter accuracy requirements (smaller $\epsilon$, larger $\eta$) lead to more restrictive rate constraints. These trends closely match the theoretical bounds derived in Theorems \ref{theo: Opt_design}-\ref{theo: Opt_ep_del}.
\begin{figure}[H]
    \centering
    \begin{minipage}[b]{0.5\textwidth}
        \centering
        \includegraphics[width=\linewidth]{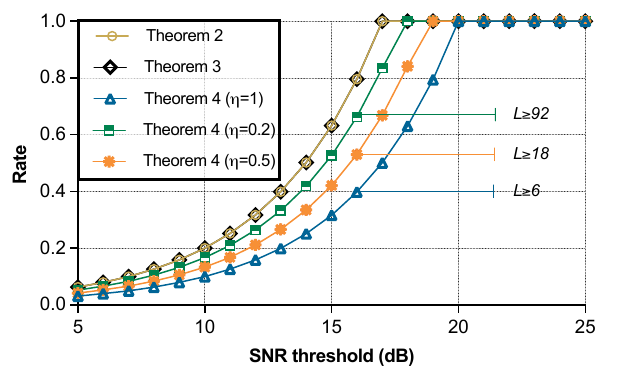}
        \vspace{-7mm}
        \caption{Rate regions of the proposed method ($\epsilon=0.02$, $\delta=0.2$). }
        \vspace{-3mm}
        \label{fig:rate-region}
    \end{minipage}
\end{figure}
\vspace{-0.5em}
\subsection{Impact of Codeword Length $\Tilde{L}$}
Fig.~\ref{fig:mse-length} displays the MSE over $500$ transmissions for different codeword lengths, setting $R=0.5$ and $\rho_X \leq 15\mathrm{dB}$. The results indicate that while the codeword length $\Tilde{L}$ presents variability, the sample mean remains nearly constant. Meanwhile, the variance decreases as the codeword length increases, indicating stronger concentration around the mean. This behavior is consistent with Theorem \ref{theo: Opt_design} and \ref{theo: asymptot_rate}, where the MSE converges to its expectation as the codeword length grows.
\begin{figure}[H]
    \centering
    \vspace{-2.5mm}
    \begin{minipage}[b]{0.5\textwidth}
        \centering
        \includegraphics[width=\linewidth]{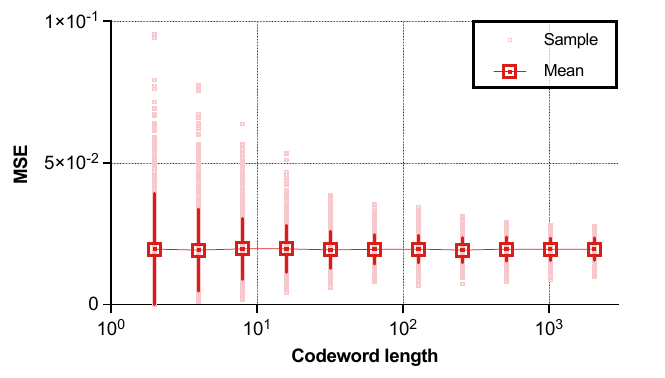}
        \vspace{-7mm}
        \caption{MSE over multiple transmissions, showing the sample mean and variance (indicated by vertical error bar).}
        \vspace{-3mm}
        \label{fig:mse-length}
    \end{minipage}
\end{figure}

\vspace{-0.5em}
\section{Conclusions}
This letter investigates channel coding tailored to analog AirComp to combat channel impairments e.g., noise and fading. By applying an identical encoding matrix across users, the proposed channel-coded AirComp preserves the aggregation structure while leveraging bandwidth expansion to reduce computation error. 
The MSE is characterized in both expectation and distribution, the optimal coding design is derived, and achievable rate regions are established under various criteria, including asymptotically infinite and finite blocklength regimes. 
Overall, this work bridges channel coding and AirComp and opens a new direction for reliable AirComp.
Future work may consider: \textit{i)} the design of optimal coding schemes under imperfect CSIT, and \textit{ii)} transceiver design for channel-coded AirComp under system impairments.



\vspace{-0.5em}
\bibliographystyle{IEEEtran.bst}
\bibliography{IEEEabrv,ref}

\vfill

\end{document}